\documentclass{article}%
\usepackage{amsmath}
\usepackage{amsfonts}
\usepackage{amssymb}
\usepackage{graphicx}%
\providecommand{\U}[1]{\protect\rule{.1in}{.1in}}
\newtheorem{theorem}{Theorem}

\newtheorem{lemma}[theorem]{Lemma}

\newtheorem{proposition}[theorem]{Proposition}
\newtheorem{remark}[theorem]{Remark}

\newenvironment{proof}[1][Proof]{\noindent\textbf{#1.} }{\ \rule{0.5em}{0.5em}}
\begin{document}

\title{Convertibility of Observables}
\author{Keiji Matsumoto\\Quantum Computation Group, National Institute of Informatics, \ \\2-1-2 Hitotsubashi, Chiyoda-ku, Tokyo 101-8430 \\e-mail : keiji@nii.ac.jp}
\maketitle

\begin{abstract}
Some problems of quantum information, cloning, estimation and testing of
states, universal coding etc., are special example of the following `state
convertibility' problem. In this paper, we consider the dual of this problem,
'observable conversion problem'. Given families of operators $\left\{
L_{\theta}\right\}  _{\theta\in\Theta}$ and $\left\{  M_{\theta}\right\}
_{\theta\in\Theta}$\ , we ask whether there is a completely positive (sub)
unital map which sends $L_{\theta}$ to $M_{\theta}$ for each $\theta$. We give
necessary and sufficient conditions for the convertibility in some special cases.

\end{abstract}

\section{Introduction}

\subsection{Problem treated in the paper}

Some problems of quantum information, cloning, estimation and testing of
states, universal coding etc., are special example of the following `state
convertibility' problem. Consider parameterized families of density operators
$\mathcal{E}=\left\{  \rho_{\theta}\right\}  _{\theta\in\Theta}$,
$\,\,\mathcal{F}=\,\left\{  \sigma_{\theta}\right\}  _{\theta\in\Theta}$, each
on a finite dimensional Hilbert space $\mathcal{H}$ and $\mathcal{K}$,
respectively. Then the question is whether there is a completely trace
preserving positive (CPTP) map $\Phi$ such that
\begin{equation}
\forall\theta\in\Theta,\,\left\Vert \Phi\left(  \rho_{\theta}\right)
-\sigma_{\theta}\right\Vert _{1}\leq e_{\theta}, \label{convert-state-e}%
\end{equation}
where $e_{\theta}$ are small non-negative numbers. Its errorless version is to
find whether there is a CPTP map $\Lambda$ such that%
\begin{equation}
\forall\theta\in\Theta,\,\,\Phi\left(  \rho_{\theta}\right)  =\sigma_{\theta}.
\label{convert-state}%
\end{equation}

In this paper we consider its `dual' problem, or `observable convertibility'
problem. Denote by $\mathcal{L}\left(  \mathcal{H}\right)  $ the set of all
the linear operators over $\mathcal{H}$ ($\dim\mathcal{H\,<\infty}$ unless
otherwise mentioned) and $I_{\mathcal{H}}$ is the identity element of
$\mathcal{L}\left(  \mathcal{H}\right)  $. Consider sets of positive
operators,
\[
\mathcal{\hat{E}}=\left\{  L_{\theta}\right\}  _{\theta\in\Theta
},\,\,\mathcal{\hat{F}}=\,\left\{  M_{\theta}\right\}  _{\theta\in\Theta},
\]
where $\left\vert \Theta\right\vert <\infty$ ($\Theta=\left\{  1,\cdots
,\left\vert \Theta\right\vert \right\}  $) and $L_{\theta}$ and $M_{\theta}$
is acting on $d$-dimensional space $\mathcal{H}$ and $d^{\prime}$-dimensional
space $\mathcal{K}$. respectively. Our question is whether there is a complete
positive (CP) (sub)unital or unital map $\Lambda$ such that \
\begin{equation}
\Lambda\left(  L_{\theta}\right)  =M_{\theta},\,\forall\theta\in\Theta.
\label{convert}%
\end{equation}

\subsection{Motivation}

One application of the problem treated here is the question of the order
structure of POVMs treated in \cite{Bucemietal}: If POVM $\left\{
E_{i}\right\}  _{i\in I}$ and $\left\{  F_{i}\right\}  _{i\in I}$ \ satisfies
\[
\forall i\in I,\,\Lambda\left(  E_{i}\right)  =F_{i}%
\]
for certain CP unital map $\Lambda$, the latter can be made from the former by
a physical transformation. Thus $\left\{  E_{i}\right\}  _{i\in I}$ is more
useful than $\left\{  F_{i}\right\}  _{i\in I}$ for any tasks, obviously. To
check the relation holds, one instead can check (\ref{convert}) by setting
$\Theta=I\backslash\left\{  i_{0}\right\}  $, and $L_{\theta}=c_{\theta
}E_{\theta}$, $M_{\theta}=c_{\theta}F_{\theta}$, where $c_{\theta}%
\in\mathbb{R}$. In this case, $\Lambda$ considered is a CP unital map
$\Lambda$.

Sometimes, we are interested in sub-POVMs, or sets of positive operators with
\[
\sum_{\theta\in\Theta}L_{\theta}\leq I_{\mathcal{H}}\text{.}%
\]
For example, in case of detection of unknown states $\{\rho_{\theta}%
\}_{\theta\in\Theta}$ , sometimes we allow the answer "I don't know". Then, if
$L_{\theta}$ corresponds to the answer `$\rho_{\theta}$ is the true state',
the sum $\sum_{\theta\in\Theta}L_{\theta}$ is smaller than or equal to
$I_{\mathcal{H}}$, and $I_{\mathcal{H}}-\sum_{\theta\in\Theta}L_{\theta}$
corresponds to "I don't know". In this case, transformation by a CP subunital
map is of interest. Suppose (\ref{convert}) holds for a CP subunital map
$\Lambda$ and $\sum_{\theta\in\Theta}M_{\theta}$ is smaller than or equal to
$I_{\mathcal{K}}$ $.$Then the measurement corresponding to $\left\{
M_{\theta}\right\}  _{\theta\in\Theta}$ is realized by the one corresponding
to $\left\{  L_{\theta}\right\}  _{\theta\in\Theta}$ in the following manner.
Given an input state $\rho$, we perform the measurement
\[
\rho\rightarrow\left\{
\begin{array}
[c]{cc}%
\sqrt{I_{\mathcal{K}}-\Lambda^{\ast}\left(  I_{\mathcal{H}}\right)  }\rho
\sqrt{I_{\mathcal{K}}-\Lambda^{\ast}\left(  I_{\mathcal{H}}\right)  }, &
\text{output}=\text{`I don't know'},\\
\Lambda\left(  \rho\right)  , & \text{output}=\text{`proceed'}.
\end{array}
\right.
\]
If the measurement result is `proceed' we apply the measurement corresponding
to $\left\{  L_{\theta}\right\}  _{\theta\in\Theta}$.

Also, (\ref{convert}) is related to the `state conversion' problem. Let us
define $S\in\mathcal{L}\left(  \mathcal{H}\right)  $ and $T\in\mathcal{L}%
\left(  \mathcal{K}\right)  $ by
\begin{equation}
\sum_{\theta\in\Theta}\rho_{\theta}=SS^{\dagger},\,\sum_{\theta\in\Theta
}\sigma_{\theta}=TT^{\dagger},\, \label{rho=SS}%
\end{equation}
and define
\begin{equation}
L_{\theta}:=S^{-1}\rho_{\theta}S^{\dagger-1},\,M_{\theta}:=T^{-1}%
\sigma_{\theta}T^{\dagger-1}. \label{rho=SLS}%
\end{equation}
(Note that in `state conversion' problem, we can suppose $\mathrm{supp}%
\,S=\mathcal{H}$ and $\mathrm{supp}\,T=\mathcal{K}$ without loss of
generality.) If (\ref{convert-state}) holds for a CP trace preserving map
$\Phi$, the map
\begin{equation}
\Lambda\left(  X\right)  :=T^{-1}\Phi\left(  SX\,S^{\dagger}\right)
T^{\dagger-1} \label{Lambda-0}%
\end{equation}
is CP and unital, and $\Lambda$ satisfies (\ref{convert}) and
\begin{equation}
\Lambda^{\ast}\left(  T^{\dagger}T\right)  =S^{\dagger}S. \label{L(TT)=SS}%
\end{equation}
So existence of a CP unital map with (\ref{convert}) is necessary condition
for existence of a CPTP map with (\ref{convert-state}). \ 

Conversely, if (\ref{convert}) holds\ for a CP unital map $\Lambda$, the CP
map $\Phi$ defined by
\begin{equation}
\Phi\left(  X\right)  :=T\Lambda\left(  S^{-1}\,X\,S^{\dagger-1}\right)
T^{\dagger} \label{Phi-def}%
\end{equation}
satisfies (\ref{convert-state}). $\Phi$ is trace preserving if and only if
(\ref{L(TT)=SS}) holds.

Another link to state convertibility problem is as follows. It is known that
the existence of CPTP map with (\ref{convert-state}) is equivalent to, when
$\left\vert \Theta\right\vert <\infty$,
\[
\inf_{\Lambda}\,\sum_{\theta\in\Theta}\mathrm{tr}\,\Lambda\left(  L_{\theta
}\right)  \rho_{\theta}\,\,p_{\theta}\leq\sum_{\theta\in\Theta}\mathrm{tr}%
\,L_{\theta}\,\sigma_{\theta}\,p_{\theta}%
\]
holds for any parameterized family of positive operators $\left\{  L_{\theta
}\right\}  _{\theta\in\Theta}$ with $\left\Vert L_{\theta}\right\Vert \leq1$
and for any probability distributions $p_{\theta}$ on $\Theta$. Here,
$\Lambda$ moves all over the CP trace preserving maps, or all over the CP
trance non-increasing maps. To solve this problem, the knowledge about
$\left\{  \Lambda^{\ast}\left(  L_{\theta}\right)  \right\}  _{\theta\in
\Theta}$ when $\Lambda^{\ast}$ moves over all the CP (sub)unital maps will be
of some help.

\subsection{Notations, conventions, and a small technical point}

Here we add some more notations used in the paper. In this paper $d=\dim$
$\mathcal{H\,<\infty}$ unless otherwise mentioned. A map $\Lambda$ from
$\mathcal{L}\left(  \mathcal{H}\right)  $ to $\mathcal{L}\left(
\mathcal{K}\right)  $ is said to be unital if $\Lambda\left(  I_{\mathcal{H}%
}\right)  =I_{\mathcal{K}}$, and subunital if $\Lambda\left(  I_{\mathcal{H}%
}\right)  \leq I_{\mathcal{K}}$. By definition, any unital map is subunital.
$P_{\mathcal{H}}$ is the projection onto the vector space $\mathcal{H}$.
$\left\Vert A\right\Vert $ , $\lambda_{\max}\left(  A\right)  $ and
$\lambda_{\min}\left(  A\right)  $ denotes the operator norm, the largest
eigenvalue, and the smallest eigenvalue, respectively. Also, $\mathrm{sp}%
\left(  A\right)  :=\lambda_{\max}\left(  A\right)  -\lambda_{\min}\left(
A\right)  $. $\left\Vert A\right\Vert _{1}$ is the trace norm of $A$,
$\left\Vert A\right\Vert _{1}:=\mathrm{tr}\,\sqrt{A^{\dagger}A}$. For a
matrices $A=\left[  A_{i,j}\right]  $ and $B=\left[  B_{i,j}\right]  $, the
Hadamard product $A\circ B$ is defined by $\left(  A\circ B\right)
_{i,j}=A_{i,j}B_{i,j}$.

In dealing with `observable convertibility', there is a subtle point which was
absent in `state convertibility' problem. In the latter, the input Hilbert
space $\mathcal{H}$ and the output Hilbert space $\mathcal{K}$ could be any
space which contains $\sum_{\theta\in\Theta}\mathrm{supp}\,\rho_{\theta}$ and
$\sum_{\theta\in\Theta}\mathrm{supp}\,\sigma_{\theta}$, respectively. This is
not the case in case that $\Lambda$ is a unital map. The reason is as follows.

\ Let $\Lambda$ be a linear map from $\mathcal{L}\left(  \mathcal{H}\right)  $
to $\mathcal{L}\left(  \mathcal{K}\right)  $, and let $\mathcal{H}^{\prime}%
$and $\mathcal{K}^{\prime}$ be Hilbert spaces with $\mathcal{H}^{^{\prime}%
}\subset\mathcal{H}$ and $\mathcal{K}\subset\mathcal{K}^{^{\prime}}$ . \ Then
the restriction of $\Lambda$ to $\mathcal{L}\left(  \mathcal{H}^{\prime
}\right)  $ nor the imbedding the range of $\Lambda$ into $\mathcal{K}%
^{^{\prime}}$ is not unital in general. So even if there is a CP map with
(\ref{convert}) and $\Lambda\left(  I_{\mathcal{H}}\right)  =I_{\mathcal{K}}$,
there might not be any $\Lambda$ with and $\Lambda\left(  I_{\mathcal{H}%
^{\prime}}\right)  =I_{\mathcal{K}^{\prime}}$. Therefore, not only the sets of
the observables $\mathcal{\hat{E}}=\left\{  L_{\theta}\right\}  _{\theta
\in\Theta}$ and $\,\mathcal{\hat{F}}=\,\left\{  M_{\theta}\right\}
_{\theta\in\Theta}$ , the choice of underlying Hilbert spaces $\mathcal{H}$
and $\mathcal{K}$ is important part of the problem.

In dealing with problem, an easy and useful necessary condition for
(\ref{convert}) is $\Lambda\left(  L_{\theta}\right)  =M_{\theta}$ for each
$\theta\in\Theta_{0}$, where $\Theta_{0}$ is a subset of $\Theta$. In case
that $\sum_{\theta\in\Theta_{0}}\mathrm{supp}\,L_{\theta}$ is strictly smaller
than $\sum_{\theta\in\Theta}\mathrm{supp}\,L_{\theta}$, one may be tempted to
replace $\mathcal{H}$ by $\sum_{\theta\in\Theta_{0}}\mathrm{supp}\,L_{\theta}%
$. But this is not possible, as mentioned above. This is one reason why we
also pay attention to conversion by subunital map. In this case one can freely
chose underlying Hilbert space, giving tractable necessary conditions for
existence of a unital map with (\ref{convert}).

\section{An application to a `state conversion' problem}

Suppose
\[
\rho_{\theta}=\left\vert u_{\theta}\right\rangle \left\langle u_{\theta
}\right\vert ,\,\,\sigma_{\theta}=\left\vert v_{\theta}\right\rangle
\left\langle v_{\theta}\right\vert ,
\]
(In the what follows, we do not assume $\mathrm{tr}\,\rho_{\theta}%
=\mathrm{tr}\,\sigma_{\theta}=1$. Thus $u_{\theta}$ and $v_{\theta}$ may not
be normalized.) Denote by $\mathcal{U}$ and $\mathcal{V}$ the family $\left\{
u_{\theta}\right\}  _{\theta\in\Theta}$ and $\left\{  v_{\theta}\right\}
_{\theta\in\Theta}$, respectively. Denote by $G_{\mathcal{U}}$ and
$G_{\mathcal{V}}$ the Gram matrix of $\mathcal{U}$ and $\mathcal{V}$,
respectively, that is,
\[
G_{\mathcal{U},\theta,\theta^{\prime}}:=\left\langle u_{\theta}\right\vert
\left.  u_{\theta^{\prime}}\right\rangle ,\,G_{\mathcal{V},\theta
,\theta^{\prime}}:=\left\langle v_{\theta}\right\vert \left.  v_{\theta
^{\prime}}\right\rangle .
\]
\ 

\begin{theorem}
\label{th:CJW}( Theorem\thinspace2 of \cite{CJW}) There is a CP trace
preserving map $\Phi$ from $\mathcal{L}\left(  \mathcal{H}\right)  $ to
$\mathcal{L}\left(  \mathcal{K}\right)  $ satisfying (\ref{convert-state}) if
and only if there is a matrix $H=\left[  H_{\theta,\theta^{\prime}}\right]  $
such that
\begin{align}
G_{\mathcal{U},\theta,\theta^{\prime}}  &  =G_{\mathcal{V},\theta
,\theta^{\prime}}\circ H,\label{G=GH-s}\\
H  &  \geq0,\,H_{\theta,\theta^{\prime}}=1. \label{H}%
\end{align}

\end{theorem}

\begin{theorem}
\label{th:CJW-2}(Corollary 1 of \cite{CJW})There is a CP trace preserving map
$\Phi$ from $\mathcal{L}\left(  \mathcal{H}\right)  $ to $\mathcal{L}\left(
\mathcal{K}\right)  $ satisfying (\ref{convert}) and $\Phi^{\prime}$ from
$\mathcal{L}\left(  \mathcal{K}\right)  $ to $\mathcal{L}\left(
\mathcal{H}\right)  $ satisfying
\begin{equation}
\Phi^{\prime}\left(  \sigma_{\theta}\right)  =\rho_{\theta},\forall\theta
\in\Theta\label{reverse}%
\end{equation}
if and only if \ $\mathcal{U}$ and $\mathcal{V}$ are unitary equivalent.
\end{theorem}

\begin{remark}
In Theorem\thinspace2 of \cite{CJW}, they do not have condition that
$H_{\theta,\theta}=1$. In their case, they consider mapping from a family of
normalized vectors to another family of normalized vectors, so that
$H_{\theta,\theta}=1$ holds automatically. In our case, we have to impose this
additional constrain because the system of vectors may not be normalized.
\end{remark}

To show link between `state conversion' and `state conversion', we give
another proof of Theorems\thinspace\ref{th:CJW}-\ref{th:CJW-2}, in case that
$\mathcal{U}$ and $\mathcal{V}$ are linearly independent.

Suppose $\mathrm{supp}\,\sum_{\theta\in\Theta}\rho_{\theta}=\mathcal{H}$ \ and
$\mathrm{supp}\,\sum_{\theta\in\Theta}\sigma_{\theta}=\mathcal{K}$ without
loss of generality, so that $\dim\mathcal{H}=\dim\mathcal{K}=\left\vert
\Theta\right\vert $. Define $S\in\mathcal{L}\left(  \mathcal{H}\right)  $ and
$T\in\mathcal{L}\left(  \mathcal{K}\right)  $ by
\[
S=\sum_{\theta\in\Theta}\left\vert u_{\theta}\right\rangle \left\langle
e_{\theta}\right\vert ,T=\sum_{\theta\in\Theta}\left\vert v_{\theta
}\right\rangle \left\langle f_{\theta}\right\vert ,
\]
where $\left\{  e_{\theta}\right\}  _{\theta\in\Theta}$ and $\left\{
f_{\theta}\right\}  _{\theta\in\Theta}$ is a complete orthonormal basis of
$\mathcal{H}$ and $\mathcal{K}$, respectively. Then $S$ and $T$ satisfy
(\ref{rho=SS}). It is easy to check $\left\{  L_{\theta}\right\}  _{\theta
\in\Theta}$ and $\left\{  M_{\theta}\right\}  _{\theta\in\Theta}$ defined by
(\ref{rho=SLS}) are orthonormal projections,
\[
L_{\theta}=\left\vert e_{\theta}\right\rangle \left\langle e_{\theta
}\right\vert ,\,M_{\theta}=\left\vert f_{\theta}\right\rangle \left\langle
f_{\theta}\right\vert .
\]

Suppose there is a CPTP map with (\ref{convert-state}). Then, a CP map
$\Lambda$ defined by (\ref{Lambda-0}) is unital and satisfies (\ref{convert}),
so it is a `dephasing map',
\begin{equation}
\Lambda\left(  \left\vert e_{\theta}\right\rangle \left\langle e_{\theta
^{\prime}}\right\vert \right)  =H_{\theta,\theta^{\prime}}\left\vert
f_{\theta}\right\rangle \left\langle f_{\theta^{\prime}}\right\vert ,
\label{L=H}%
\end{equation}
where $H=\left[  H_{\theta,\theta^{\prime}}\right]  $ satisfies (\ref{H}). It
is easy to see
\[
\Lambda^{\ast}\left(  \left\vert f_{\theta}\right\rangle \left\langle
f_{\theta^{\prime}}\right\vert \right)  =H_{\theta,\theta^{\prime}}\left\vert
e_{\theta}\right\rangle \left\langle e_{\theta^{\prime}}\right\vert .
\]
Also $\Lambda$ satisfies (\ref{L(TT)=SS}),
\begin{align*}
\Lambda^{\ast}\left(  \sum_{\theta,\theta^{\prime}\in\Theta}G_{\mathcal{V}%
,\theta,\theta^{\prime}}\left\vert f_{\theta}\right\rangle \left\langle
f_{\theta^{\prime}}\right\vert \right)   &  =\sum_{\theta,\theta^{\prime}%
\in\Theta}H_{\theta,\theta^{\prime}}G_{\mathcal{V},\theta,\theta^{\prime}%
}\left\vert e_{\theta}\right\rangle \left\langle e_{\theta^{\prime}%
}\right\vert \\
&  =\sum_{\theta,\theta^{\prime}\in\Theta}G_{\mathcal{U},\theta,\theta
^{\prime}}\left\vert e_{\theta}\right\rangle \left\langle e_{\theta^{\prime}%
}\right\vert .
\end{align*}
Therefore, (\ref{G=GH-s}) is necessary.

Conversely, if (\ref{G=GH-s}) holds, the map defined by (\ref{L=H}) is a CP
unital map, and satisfies (\ref{convert}) and (\ref{L(TT)=SS}). Also, the CP
map defined by (\ref{Phi-def}) is trace-preserving and satisfies
(\ref{convert-state}). Thus (\ref{G=GH-s}) is sufficient. Thus we obtain
Theorems\thinspace\ref{th:CJW} in case that $\mathcal{U}$ and $\mathcal{V}$
are linearly independent.

Next, \ suppose there is a CPTP map $\Phi^{\prime}$ with (\ref{reverse})
exists. We also suppose $\left\vert e_{\theta}\right\rangle =\left\vert
f_{\theta}\right\rangle $ without loss of generality. Then, by \cite{Petz},
the map
\[
\Phi^{\prime}\left(  X\right)  :=U_{2}\Lambda^{\ast}\left(  U_{1}^{\dagger
}XU_{1}\right)  U_{2}^{\dagger}%
\]
should be an example of such a map, where $U_{1}$ and $U_{2}$ are unitary
operators defined by
\[
\sqrt{TT^{\dagger}}U_{1}=T,\,\sqrt{SS^{\dagger}}U_{2}=S.
\]
Then $\Phi^{\prime}$ maps a pure state to another pure state only if
$\Lambda^{\ast}$ does so. In turn, $\Lambda^{\ast}$ maps a pure state to
another pure state only if it does not change the pure state. Therefore, we
have Theorems\thinspace\ref{th:CJW-2} in case that $\mathcal{U}$ and
$\mathcal{V}$ are linearly independent.

\section{Conversion between rank-1 operators}

In this section,
\[
L_{\theta}=\left\vert u_{\theta}\right\rangle \left\langle u_{\theta
}\right\vert ,\,M_{\theta}=\left\vert v_{\theta}\right\rangle \left\langle
v_{\theta}\right\vert \,.
\]
Let $\mathcal{U}^{\uparrow}:=\left\{  u_{\theta}^{\uparrow}\right\}
_{\theta\in\Theta}$ and $\mathcal{V}^{\uparrow}:=\left\{  v_{\theta}%
^{\uparrow}\right\}  _{\theta\in\Theta}$ be the dual system (if exists) of
$\mathcal{U}$ and $\mathcal{V}$, respectively,%
\[
\left\langle u_{\theta}^{\uparrow}\right\vert \left.  u_{\theta^{\prime}%
}\right\rangle =\delta_{\theta,\theta^{\prime}},\left\langle v_{\theta
}^{\uparrow}\right\vert \left.  v_{\theta^{\prime}}\right\rangle
=\delta_{\theta,\theta^{\prime}}.
\]

\begin{lemma}
\label{lem:independent}Suppose that $\mathcal{V}$ is linearly independent.
Then a positive operator $C$ supported on $\mathrm{span}\mathcal{V}$ is
identical to $\left\vert v_{\theta}\right\rangle \left\langle v_{\theta
}\right\vert $ \ if and only if
\begin{equation}
\left\langle v_{\theta^{\prime}}^{\uparrow}\right\vert C\left\vert
v_{\theta^{\prime}}^{\uparrow}\right\rangle =\delta_{\theta,\theta^{\prime}}.
\label{dual-inner}%
\end{equation}
Also, there is a CP unital map $\Lambda$ satisfying (\ref{convert}). Then
$\mathcal{U}$ is also linearly independent.
\end{lemma}

\begin{proof}
If $C$ is identical to $\left\vert v_{\theta}\right\rangle \left\langle
v_{\theta}\right\vert $, (\ref{dual-inner}) is trivially true. If
(\ref{dual-inner}) holds, the positive operator $C$ has to have null space
spanned by $\left\{  v_{\theta^{\prime}}^{\uparrow};\theta^{\prime}\in
\Theta,\theta^{\prime}\neq\theta\right\}  $. Therefore, $C$ is constant
multiple of $\left\vert v_{\theta}\right\rangle \left\langle v_{\theta
}\right\vert \,$. The constant factor is fixed by the condition $\left\langle
v_{\theta^{\prime}}^{\uparrow}\right\vert C\left\vert v_{\theta}^{\uparrow
}\right\rangle =1$. Thus, we have the first assertion.

(\ref{convert}) holds only if
\begin{align*}
\mathrm{tr}\,\Lambda^{\ast}\left(  \left\vert v_{\theta_{0}}^{\uparrow
}\right\rangle \left\langle v_{\theta_{0}}^{\uparrow}\right\vert \right)
\left\vert u_{\theta}\right\rangle \left\langle u_{\theta}\right\vert  &
=\mathrm{tr}\,\left\vert v_{\theta_{0}}^{\uparrow}\right\rangle \left\langle
v_{\theta_{0}}^{\uparrow}\right\vert \Lambda\left(  \left\vert u_{\theta
}\right\rangle \left\langle u_{\theta}\right\vert \right)  =0,\,\theta
\neq\theta_{0},\\
\mathrm{tr}\,\Lambda^{\ast}\left(  \left\vert v_{\theta_{0}}^{\uparrow
}\right\rangle \left\langle v_{\theta_{0}}^{\uparrow}\right\vert \right)
\left\vert u_{\theta_{0}}\right\rangle \left\langle u_{\theta_{0}}\right\vert
&  =\mathrm{tr}\,\left\vert v_{\theta_{0}}^{\uparrow}\right\rangle
\left\langle v_{\theta_{0}}^{\uparrow}\right\vert \Lambda\left(  \left\vert
u_{\theta_{0}}\right\rangle \left\langle u_{\theta_{0}}\right\vert \right)
=1,
\end{align*}
These lead to contradiction if $u_{\theta_{0}}$ is in the span of
$\mathcal{U}$. Therefore, $\mathcal{U}$ should be linearly independent.
\end{proof}

The proof of the following lemma is almost immediate.

\begin{lemma}
There is a CP map $\Lambda$ satisfying (\ref{convert}) holds if and only if
\begin{equation}
\exists\alpha_{\theta,i},\,\,\,\alpha_{\theta,i}\left\vert v_{\theta
}\right\rangle =W_{i}\left\vert u_{\theta}\right\rangle ,\,\sum_{i}\left\vert
\alpha_{\theta,i}\right\vert ^{2}=1, \label{v=wu}%
\end{equation}
where $W_{i}$'s are Kraus operators of $\Lambda$.
\end{lemma}

The following theorem is almost a dual of Theorem\thinspace\ref{th:CJW}.

\begin{theorem}
\label{th:1-dim-ineq}Suppose that $\mathcal{V}$ is linearly independent. Then
there is a CP subunital map $\Lambda^{\ast}$ satisfying (\ref{convert}) if and
only if there is a matrix $H=\left[  H_{\theta,\theta^{\prime}}\right]  $ such that
\end{theorem}

\begin{align}
G_{\mathcal{V}}^{-1}  &  \geq H\circ G_{\mathcal{U}}^{-1},\label{G>HG}\\
H  &  \geq0,\,H_{\theta,\theta}=1,\,\theta\in\Theta, \label{H-2}%
\end{align}

\begin{proof}
By Lemma\thinspace\ref{lem:independent}, $\mathcal{U=}\left\{  \left\vert
u_{\theta}\right\rangle \right\}  _{\theta\in\Theta}$ is also linearly
independent. Without loss of generality, we suppose $\mathcal{H}%
=\mathrm{span}\,\mathcal{U}$\thinspace and $\mathcal{K}=\mathrm{span}%
\,\mathcal{V}$, and thus $d=d^{\prime}=\left\vert \Theta\right\vert $.
\thinspace Let $\left\{  W_{i}\right\}  $ be Kraus operators of $\Lambda$,
$\Lambda\left(  L\right)  =\sum_{i}W_{i}LW_{i}^{\dagger}$.

Let us denote by $[\mathcal{U}]$ and $[\mathcal{V}]$ \ the matrix whose
$\theta$'s column vector is $\left\vert u_{\theta}\right\rangle $ and
$\left\vert v_{\theta}\right\rangle $, respectively. Observe they are square
matrices and invertible. So the condition (\ref{v=wu}) is rewritten as
\[
\lbrack\mathcal{V}]\mathrm{diag}\left(  \alpha_{1,i},\cdots\alpha_{\left\vert
\Theta\right\vert ,i}\right)  =W_{i}[\mathcal{U}],\,
\]
or equivalently%
\begin{equation}
W_{i}=[\mathcal{V}]\mathrm{diag}\left(  \alpha_{1,i},\cdots\alpha_{\left\vert
\Theta\right\vert ,i}\right)  [\mathcal{U}]^{-1}. \label{W-def}%
\end{equation}
A CP map $\Lambda^{\ast}$ given by the Kraus operators (\ref{W-def}) is
subunital if and only if
\begin{align*}
I_{\mathcal{K}}  &  \geq\sum_{i}W_{i}W_{i}^{\dagger}=\sum_{i}[\mathcal{V}%
]\mathrm{diag}\left(  \alpha_{1,i},\cdots\alpha_{\left\vert \Theta\right\vert
,i}\right)  [\mathcal{U}]^{-1}[\mathcal{U}]^{-1\dagger}\mathrm{diag}\left(
\overline{\alpha_{1,i}},\cdots\overline{\alpha_{\left\vert \Theta\right\vert
,i}}\right)  [\mathcal{V}]^{\dagger}\\
&  =\sum_{i}[\mathcal{V}]\mathrm{diag}\left(  \alpha_{1,i},\cdots
\alpha_{\left\vert \Theta\right\vert ,i}\right)  G_{\mathcal{U}}%
^{-1}\mathrm{diag}\left(  \overline{\alpha_{1,i}},\cdots\overline
{\alpha_{\left\vert \Theta\right\vert ,i}}\right)  [\mathcal{V}]^{\dagger}\\
&  =[\mathcal{V}]H\circ G_{\mathcal{U}}^{-1}[\mathcal{V}]^{\dagger},
\end{align*}
where $H_{\theta,\theta^{\prime}}:=\sum_{i}\alpha_{\theta,i}\overline
{\alpha_{\theta^{\prime},i}}$. Since $H$ satisfies (\ref{H-2}) and
$[\mathcal{V}]^{-1}[\mathcal{V}]^{-1\dagger}=G_{\mathcal{V}}^{-1}$, we have
the assertion.
\end{proof}

The following theorem can be proved in almost parallel manner as the previous
theorem. But to make the relation with Theorem\thinspace\ref{th:CJW}, we give
the proof using Theorem\thinspace\ref{th:CJW}.

\begin{theorem}
\label{th:1-dim}Suppose that $\mathcal{V}$ is linearly independent, and
$\mathrm{span}\,\mathcal{V=K}$, $\mathrm{span}\,\mathcal{U=H}$. Then there is
a CP unital map $\Lambda^{\ast}$ satisfying (\ref{convert}) if and only if
there is a matrix $H=\left[  H_{\theta,\theta^{\prime}}\right]  $ with
(\ref{H-2}) and
\begin{equation}
G_{\mathcal{V}}^{-1}=H\circ G_{\mathcal{U}}^{-1}. \label{G=HG}%
\end{equation}

\end{theorem}

\begin{proof}
Since by Lemma\thinspace\ref{lem:independent}, $\left\{  \left\vert u_{\theta
}\right\rangle \right\}  _{\theta\in\Theta}$ is also linearly independent,
$\left\vert \Theta\right\vert =d=d^{\prime}$. \ By Lemma\thinspace
\ref{lem:independent}, (\ref{convert}) is same as
\[
\delta_{\theta,\theta^{\prime}}=\mathrm{tr}\,\left\vert v_{\theta^{\prime}%
}^{\uparrow}\right\rangle \left\langle v_{\theta^{\prime}}^{\uparrow
}\right\vert \Lambda\left(  \left\vert u_{\theta}\right\rangle \left\langle
u_{\theta}\right\vert \right)  =\mathrm{tr}\,\Lambda^{\ast}\left(  \left\vert
v_{\theta^{\prime}}^{\uparrow}\right\rangle \left\langle v_{\theta^{\prime}%
}^{\uparrow}\right\vert \right)  \left\vert u_{\theta}\right\rangle
\left\langle u_{\theta}\right\vert .
\]
Therefore, replacing $v_{\theta}$ and $v_{\theta}^{\uparrow}$ in
(\ref{dual-inner}) by $u_{\theta}$ and $u_{\theta}^{\uparrow}$ respectively,
we have%
\[
\Lambda^{\ast}\left(  \left\vert v_{\theta^{\prime}}^{\uparrow}\right\rangle
\left\langle v_{\theta^{\prime}}^{\uparrow}\right\vert \right)  =\left\vert
u_{\theta^{\prime}}^{\uparrow}\right\rangle \left\langle u_{\theta^{\prime}%
}^{\uparrow}\right\vert .
\]
Thus, we can use Theorem\thinspace\ref{th:CJW}. Noticing $G_{\mathcal{U}%
^{\uparrow}}=G_{\mathcal{U}}^{-1}$ and $G_{\mathcal{V}^{\uparrow}%
}=G_{\mathcal{V}}^{-1}$, we obtain the asserted condition.
\end{proof}

Below, $\left[  G_{\mathcal{U}}\right]  _{\Theta_{1},\Theta_{2}}$ means
submatrix of $G_{\mathcal{U}}$ that corresponds to the rows with index in
$\Theta_{1}$ and the columns with index in $\Theta_{2}$. Also, $\mathcal{U}%
_{\Theta_{1}}$, $\mathcal{V}_{\Theta_{1}}$, $\mathcal{\hat{E}}_{\Theta_{1}}$,
and $\mathcal{\hat{F}}_{\Theta_{1}}$ means restriction of the range of
$\theta$ to $\Theta_{1}$ of $\mathcal{U}$, $\mathcal{V}$, $\mathcal{\hat{E}}$,
and $\mathcal{\hat{F}}$.

\begin{lemma}
\label{lem:detG>detG}Suppose that $\mathcal{U}$ and $\mathcal{V}$ satisfy all
the conditions of Theorem\thinspace\ref{th:1-dim-ineq}. Then, there is a CP
unital map $\Lambda$ satisfying (\ref{convert}) only if
\begin{equation}
\det\left[  G_{\mathcal{U}}\right]  _{\Theta_{1},\Theta_{1}}\geq\det\left[
G_{\mathcal{V}}\right]  _{\Theta_{1},\Theta_{1}} \label{detG>detG}%
\end{equation}
for any $\Theta_{1}\subset\Theta$.
\end{lemma}

\begin{proof}
We show the assertion using induction about $\left\vert \Theta\right\vert $.
When $\left\vert \Theta\right\vert =1$,
\[
\det\left[  G_{\mathcal{U}}\right]  =\left\Vert \left\vert u_{1}\right\rangle
\left\langle u_{1}\right\vert \right\Vert _{\infty}.
\]
Thus by (\ref{infty-monotone}), we have the assertion.

Next, suppose that the assertion is true for $\Theta_{1}$ such that
$\left\vert \Theta_{1}\right\vert \leq\left\vert \Theta\right\vert -1$.
(Without loss of generality, $\Theta_{1}\subset\Theta$.) Observe that the any
subfamily $\mathcal{U}_{\Theta_{1}}$ and $\mathcal{V}_{\Theta_{1}}$ satisfy
all the conditions of Theorem\thinspace\ref{th:1-dim-ineq}. (Here we replace
$\Theta$ in (\ref{convert}) with $\Theta_{1}$.) \ Then by the hypothesis of
the induction, (\ref{detG>detG}) holds for any subset $\Theta_{1}\neq\Theta$.
So it remains to show (\ref{detG>detG}) for $\Theta_{1}=\Theta$. Since
\[
\left(  G_{\mathcal{U}}^{-1}\right)  _{\theta,\theta}=\frac{\det\left[
G_{\mathcal{U}}\right]  _{\Theta_{2},\Theta_{2}}\,}{\det\,G_{\mathcal{U}}},
\]
where $\Theta_{2}=\Theta\backslash\left\{  \theta\right\}  $,
Theorem\thinspace\ref{th:1-dim-ineq} implies that
\[
\frac{\det\left[  G_{\mathcal{V}}\right]  _{\Theta_{2},\Theta_{2}}\,}%
{\det\,G_{\mathcal{V}}}\geq\frac{\det\left[  G_{\mathcal{U}}\right]
_{\Theta_{2},\Theta_{2}}\,}{\det\,G_{\mathcal{U}}}.
\]
Since $\det\left[  G_{\mathcal{V}}\right]  _{\Theta_{2},\Theta_{2}}\,\leq
\det\left[  G_{\mathcal{U}}\right]  _{\Theta_{2},\Theta_{2}}\,$\ by the
hypothesis of the induction, we have $\det\,G_{\mathcal{U}}\geq\det
\,G_{\mathcal{V}}$. Thus we have the assertion.
\end{proof}

\begin{theorem}
Suppose that $\mathcal{U}$ and $\mathcal{V}$ satisfy all the conditions of
Theorem\thinspace\ref{th:1-dim-ineq} and $\left\vert \Theta\right\vert =2$.
Then, there is a CP unital map $\Lambda^{\ast}$ satisfying (\ref{convert}) if
and only if
\begin{align}
\frac{\left\Vert v_{2}\right\Vert ^{2}}{\det\,G_{\mathcal{V}}}-\frac
{\left\Vert u_{2}\right\Vert ^{2}}{\det\,G_{\mathcal{U}}}  &  \geq
0,\frac{\left\Vert v_{1}\right\Vert ^{2}}{\det\,G_{\mathcal{V}}}%
-\frac{\left\Vert u_{1}\right\Vert ^{2}}{\det\,G_{\mathcal{U}}}\geq
0,\label{|theta|=2-1}\\
\left(  \frac{\left\vert \left\langle v_{1}\right\vert \left.  v_{2}%
\right\rangle \right\vert }{\det\,G_{\mathcal{V}}}-\frac{\left\vert
\left\langle u_{1}\right\vert \left.  u_{2}\right\rangle \right\vert }%
{\det\,G_{\mathcal{U}}}\right)  ^{2}  &  \leq\left(  \frac{\left\Vert
v_{2}\right\Vert ^{2}}{\det\,G_{\mathcal{V}}}-\frac{\left\Vert u_{2}%
\right\Vert ^{2}}{\det\,G_{\mathcal{U}}}\right)  \left(  \frac{\left\Vert
v_{1}\right\Vert ^{2}}{\det\,G_{\mathcal{V}}}-\frac{\left\Vert u_{1}%
\right\Vert ^{2}}{\det\,G_{\mathcal{U}}}\right)  . \label{|theta=2|-2}%
\end{align}

\end{theorem}

\begin{proof}
By Theorem \ref{th:1-dim-ineq}, existence of a CP unital map $\Lambda$ with
(\ref{convert}) is equivalent to the existence of a complex number $\eta$
\ with $\left\vert \eta\right\vert \leq1$ and
\[
\left[
\begin{array}
[c]{cc}%
\frac{\left\Vert v_{2}\right\Vert ^{2}}{\det\,G_{\mathcal{V}}}-\frac
{\left\Vert u_{2}\right\Vert ^{2}}{\det\,G_{\mathcal{U}}} & -\left(
\frac{\left\langle v_{1}\right\vert \left.  v_{2}\right\rangle }%
{\det\,G_{\mathcal{V}}}-\eta\frac{\left\langle u_{1}\right\vert \left.
u_{2}\right\rangle }{\det\,G_{\mathcal{U}}}\right) \\
-\overline{\left(  \frac{\left\langle v_{1}\right\vert \left.  v_{2}%
\right\rangle }{\det\,G_{\mathcal{V}}}-\eta\frac{\left\langle u_{1}\right\vert
\left.  u_{2}\right\rangle }{\det\,G_{\mathcal{U}}}\right)  } & \frac
{\left\Vert v_{1}\right\Vert ^{2}}{\det\,G_{\mathcal{V}}}-\frac{\left\Vert
u_{1}\right\Vert ^{2}}{\det\,G_{\mathcal{U}}}%
\end{array}
\right]  \geq0,
\]
or equivalently, \ (\ref{|theta|=2-1}) and
\[
\left\vert \frac{\left\langle v_{1}\right\vert \left.  v_{2}\right\rangle
}{\det\,G_{\mathcal{V}}}-\eta\frac{\left\langle u_{1}\right\vert \left.
u_{2}\right\rangle }{\det\,G_{\mathcal{U}}}\right\vert ^{2}\leq\left(
\frac{\left\Vert v_{2}\right\Vert ^{2}}{\det\,G_{\mathcal{V}}}-\frac
{\left\Vert u_{2}\right\Vert ^{2}}{\det\,G_{\mathcal{U}}}\right)  \left(
\frac{\left\Vert v_{1}\right\Vert ^{2}}{\det\,G_{\mathcal{V}}}-\frac
{\left\Vert u_{1}\right\Vert ^{2}}{\det\,G_{\mathcal{U}}}\right)  .
\]
Maximizing the LHS of the this expression moving $\eta$ so that $\left\vert
\eta\right\vert \leq1$, \ we obtain (\ref{|theta=2|-2}).
\end{proof}

\bigskip

\begin{lemma}
\label{lem:LinM}Suppose that \ $\left\Vert u_{\theta}\right\Vert =\left\Vert
v_{\theta}\right\Vert $ for each $\theta\in\Theta$. Then there is a CP
subunital map $\Lambda$ satisfying (\ref{convert}) exists only if
\[
L_{\theta}=\left\vert u_{\theta}\right\rangle \left\langle u_{\theta
}\right\vert \in\mathcal{M}_{\Lambda}%
\]
for each $\theta\in\Theta$.
\end{lemma}

\begin{proof}
Since $\Lambda$ is subunital, its norm does not exceeds 1 by Lemma\thinspace
\ref{lem:L-norm}. Since
\[
\frac{\left\Vert \Lambda\left(  \left\vert u_{\theta}\right\rangle
\left\langle u_{\theta}\right\vert \right)  \right\Vert }{\left\Vert
\left\vert u_{\theta}\right\rangle \left\langle u_{\theta}\right\vert
\right\Vert }=\frac{\left\Vert \left\vert v_{\theta}\right\rangle \left\langle
v_{\theta}\right\vert \right\Vert }{\left\Vert \left\vert u_{\theta
}\right\rangle \left\langle u_{\theta}\right\vert \right\Vert },
\]
$\left\Vert \Lambda\right\Vert =1$. Hence, since
\begin{align*}
\Lambda\left(  \left\vert u_{\theta}\right\rangle \left\langle u_{\theta
}\right\vert \cdot\left\vert u_{\theta}\right\rangle \left\langle u_{\theta
}\right\vert \right)   &  =\left\Vert u_{\theta}\right\Vert ^{2}\Lambda\left(
\left\vert u_{\theta}\right\rangle \left\langle u_{\theta}\right\vert \right)
\\
&  =\left\Vert u_{\theta}\right\Vert ^{2}\left\vert v_{\theta}\right\rangle
\left\langle v_{\theta}\right\vert =\left\Vert v_{\theta}\right\Vert
^{2}\left\vert v_{\theta}\right\rangle \left\langle v_{\theta}\right\vert \\
&  =\left\vert v_{\theta}\right\rangle \left\langle v_{\theta}\right\vert
\cdot\left\vert v_{\theta}\right\rangle \left\langle v_{\theta}\right\vert ,
\end{align*}
$\left\vert u_{\theta}\right\rangle \left\langle u_{\theta}\right\vert $ is an
element of $\mathcal{M}_{\Lambda}$ for each $\theta\in\Theta$.
\end{proof}

\begin{theorem}
\label{th:equi-length}Suppose that \ $\left\Vert u_{\theta}\right\Vert
=\left\Vert v_{\theta}\right\Vert $ for each $\theta\in\Theta$. Then there is
a CP subunital map $\Lambda$ satisfying (\ref{convert}) exists if and only if
they are unitary equivalent.
\end{theorem}

\begin{proof}
By Lemma\thinspace\ref{lem:LinM},
\begin{align}
\left\langle u_{\theta}\right\vert \left.  u_{\theta^{\prime}}\right\rangle
\Lambda\left(  \left\vert u_{\theta}\right\rangle \left\langle u_{\theta
^{\prime}}\right\vert \right)   &  =\Lambda\left(  \left\vert u_{\theta
}\right\rangle \left\langle u_{\theta}\right\vert \cdot\left\vert
u_{\theta^{\prime}}\right\rangle \left\langle u_{\theta^{\prime}}\right\vert
\right)  =\Lambda\left(  \left\vert u_{\theta}\right\rangle \left\langle
u_{\theta}\right\vert \right)  \Lambda\left(  \left\vert u_{\theta^{\prime}%
}\right\rangle \left\langle u_{\theta^{\prime}}\right\vert \right) \nonumber\\
&  =\left\vert v_{\theta}\right\rangle \left\langle v_{\theta}\right\vert
\cdot\left\vert v_{\theta^{\prime}}\right\rangle \left\langle v_{\theta
^{\prime}}\right\vert =\left\langle v_{\theta}\right\vert \left.
v_{\theta^{\prime}}\right\rangle \left\vert v_{\theta}\right\rangle
\left\langle v_{\theta^{\prime}}\right\vert , \label{uuLuu=vvvv}%
\end{align}
Also, (\ref{v=wu}) leads to
\begin{equation}
\Lambda\left(  \left\vert u_{\theta}\right\rangle \left\langle u_{\theta
^{\prime}}\right\vert \right)  =\sum_{i}\alpha_{\theta,i}\overline
{\alpha_{\theta^{\prime},i}}\left\vert v_{\theta}\right\rangle \left\langle
v_{\theta^{\prime}}\right\vert =H_{\theta,\theta^{\prime}}\left\vert
v_{\theta}\right\rangle \left\langle v_{\theta^{\prime}}\right\vert .
\label{Luu=Hvv}%
\end{equation}
Inserting (\ref{Luu=Hvv}) to (\ref{uuLuu=vvvv}) and equating the coefficients,
we have
\begin{equation}
\left\langle u_{\theta}\right\vert \left.  u_{\theta^{\prime}}\right\rangle
H_{\theta,\theta^{\prime}}=\left\langle v_{\theta}\right\vert \left.
v_{\theta^{\prime}}\right\rangle . \label{uuH=vv}%
\end{equation}

Suppose
\[
\left\langle u_{\theta}\right\vert \left.  u_{\theta^{\prime}}\right\rangle
\neq0.
\]
Then, by Lemma\thinspace\ref{lem:multiplicative},
\[
\left\vert u_{\theta}\right\rangle \left\langle u_{\theta^{\prime}}\right\vert
=\frac{1}{\left\langle u_{\theta}\right\vert \left.  u_{\theta^{\prime}%
}\right\rangle }\left\vert u_{\theta}\right\rangle \left\langle u_{\theta
}\right\vert \cdot\left\vert u_{\theta^{\prime}}\right\rangle \left\langle
u_{\theta^{\prime}}\right\vert \in\mathcal{M}_{\Lambda}.
\]
Therefore,%
\begin{align*}
\left\langle u_{\theta}\right\vert \left.  u_{\theta^{\prime}}\right\rangle
\left\vert v_{\theta}\right\rangle \left\langle v_{\theta}\right\vert  &
=\left\langle u_{\theta}\right\vert \left.  u_{\theta^{\prime}}\right\rangle
\Lambda^{\ast}\left(  \left\vert u_{\theta}\right\rangle \left\langle
u_{\theta}\right\vert \right) \\
&  =\Lambda\left(  \left\vert u_{\theta}\right\rangle \left\langle u_{\theta
}\right\vert \cdot\left\vert u_{\theta^{\prime}}\right\rangle \left\langle
u_{\theta}\right\vert \right)  =\Lambda\left(  \left\vert u_{\theta
}\right\rangle \left\langle u_{\theta}\right\vert \right)  \Lambda^{\ast
}\left(  \left\vert u_{\theta^{\prime}}\right\rangle \left\langle u_{\theta
}\right\vert \right) \\
&  =\left\vert v_{\theta}\right\rangle \left\langle v_{\theta}\right\vert
\Lambda\left(  \left\vert u_{\theta^{\prime}}\right\rangle \left\langle
u_{\theta}\right\vert \right)  .
\end{align*}
Inserting (\ref{Luu=Hvv}) into the above equation and equating the
coefficients, we obtain
\begin{equation}
\left\langle u_{\theta}\right\vert \left.  u_{\theta^{\prime}}\right\rangle
=\left\langle v_{\theta}\right\vert \left.  v_{\theta^{\prime}}\right\rangle
H_{\theta^{\prime},\theta}. \label{uu=vvH}%
\end{equation}
On the other hand, if $\left\langle u_{\theta}\right\vert \left.
u_{\theta^{\prime}}\right\rangle =0$, $\left\langle v_{\theta}\right\vert
\left.  v_{\theta^{\prime}}\right\rangle =0$\ by (\ref{uuH=vv}). Thus,
(\ref{uu=vvH}) holds in this case, too.

By Theorem\thinspace\ref{th:CJW}, (\ref{uuH=vv}) and (\ref{uu=vvH}) means that
$\mathcal{U}$ and $\mathcal{V}$ are convertible by CP trace preserving maps
back and forth. Therefore, by Theorem\thinspace\ref{th:CJW}, we have the assertion.
\end{proof}

Using theory of operator algebra more intensively, we give another proof of
Theorem\thinspace\ref{th:equi-length} below. Denote by $\left[  \left[
\mathcal{\hat{E}}\right]  \right]  $ and $\left[  \left[  \mathcal{\hat{F}%
}\right]  \right]  $ the *-algebra generated by $\mathcal{\hat{E}=}\left\{
L_{\theta}\right\}  _{\theta\in\Theta}$ and $\mathcal{\hat{F}=}\left\{
M_{\theta}\right\}  _{\theta\in\Theta}$, respectively. Since each of them is a
finite dimensional representation of a finite dimensional $C^{\ast}$-algebra,
by Lemma\thinspace\ref{lem:finite-dim}, for some unitary operators $U_{1}$ and
$U_{2}$,
\begin{align*}
\left[  \left[  \mathcal{\hat{E}}\right]  \right]   &  =U_{1}\bigoplus
_{n\in\mathbb{N}}\mathcal{L}\left(  \mathbb{C}^{n}\right)  \otimes I_{d_{1,n}%
}U_{1}^{\dagger}\,\,,\\
\left[  \left[  \mathcal{\hat{F}}\right]  \right]  \,  &  =U_{2}%
\,\bigoplus_{n\in\mathbb{N}}\mathcal{L}\left(  \mathbb{C}^{n}\right)  \otimes
I_{d_{2,n}}U_{2}^{\dagger}.
\end{align*}
Here, we used the convention that $d_{1,n}=0$ in $\left[  \left[
\mathcal{\hat{E}}\right]  \right]  $ does not have a component isomorphic to
$\mathcal{L}\left(  \mathbb{C}^{n}\right)  $. $P_{1,n}$ and $P_{2,n}$ denotes
the projection onto the subspace \ $U_{1}\left(  \mathbb{C}^{n}\otimes
I_{d_{1,n}}\right)  $ and $U_{2}\left(  \mathbb{C}^{n}\otimes I_{d_{2,n}%
}\right)  $, respectively.

\begin{proof}
Since each $L_{\theta}$ is a rank-1 operator, $d_{1,n}$ is $0$ or 1, for all
$n$. So
\[
\left[  \left[  \mathcal{\hat{E}}\right]  \right]  =U_{1}\bigoplus
_{n\in\mathbb{N},d_{1,n}\neq0}\mathcal{L}\left(  \mathbb{C}^{n}\right)
U_{1}^{\dagger},\,\,
\]
By (\ref{convert}) and Lemmas \thinspace\ref{lem:LinM}, $\Lambda^{\ast}$ is
*-homomorphism from $\left[  \left[  \mathcal{\hat{E}}\right]  \right]  $ onto
$\left[  \left[  \mathcal{\hat{F}}\right]  \right]  \,$, or is a
representation of $\left[  \left[  \mathcal{\hat{E}}\right]  \right]  $.
Hence, by Lemma\thinspace\ref{lem:finite-dim},

By Lemma\thinspace\ref{lem:LinM} and (\ref{convert}), $\Lambda^{\ast}$ is
*-homomorphism from $\left[  \left[  \mathcal{\hat{E}}\right]  \right]  $ onto
$\left[  \left[  \mathcal{\hat{F}}\right]  \right]  \,$, and thus, it is a
representation of $\left[  \left[  \mathcal{\hat{E}}\right]  \right]  $ on the
finite dimensional Hilbert space $\mathcal{K}$. Therefore, by Lemma\thinspace
\ref{lem:finite-dim}, the restriction of $\Lambda^{\ast}$ to $\left[  \left[
\mathcal{\hat{E}}\right]  \right]  $ is unitary equivalent to
\[
\bigoplus_{n\in\mathbb{N},d_{1,n}\neq0}\mathbf{I}_{\mathbb{C}^{n}}^{\left(
d_{2,n}\right)  }\text{.}%
\]
In fact, as is shown below, $d_{2,n}\neq0$ if $d_{1,n}\neq0$.

Suppose there is an $n_{0}$ such that \ $d_{1,n_{0}}\neq0$ and $d_{2,n}=0$.
Observe that $\mathrm{rank}\,L_{\theta}=1$ for all $\theta$ implies that
$P_{1,n}L_{\theta}P_{1,n}\neq0$ holds only for a single $n$ for each $\theta$.
Therefore, there is at least one $\theta$ such that $L_{\theta}\in$
$U_{1}\mathcal{L}\left(  \mathbb{C}^{n_{0}}\right)  U_{1}^{\dagger}$. Since
$n_{0}\notin N_{\Lambda^{\ast}}$ means $\Lambda^{\ast}\left(  L_{\theta
}\right)  =M_{\theta}=0$, we have contradiction. Therefore, $d_{2,n}\neq0$ if
$d_{1,n}\neq0$.

Finally, observe that $d_{2,n}$ is 0 or 1 by the same reason as $d_{1,n}$ is 0
or 1. This means that restriction of $\Lambda^{\ast}$ to $\mathcal{\hat{E}}$
is unitary equivalent to the identity operator. Therefore, $\mathcal{\hat{E}}$
and $\mathcal{\hat{F}}$ are unitary equivalent.
\end{proof}

\section{Projectors}

Modifying the second proof of the above theorem slightly, we obtain similar
result for the case where $\mathrm{rank}\,L_{\theta}=1$ and $M_{\theta}$ is a
constant multiple of a projector.

Let us divide $\mathcal{\hat{E}}$ into $\bigcup_{\kappa}\mathcal{\hat{E}%
}_{\kappa}$, so that the following conditions are satisfied;there is a
sequence $\theta=\theta_{1}$, $\theta_{2}$,$\cdots$,$\theta_{k}=\theta
^{\prime}$ with $L_{\theta_{i}}L_{\theta_{i+1}}\neq0$, $i=1,\cdots,k$ for any
$L_{\theta},\,L_{\theta^{\prime}}\in\mathcal{\hat{E}}_{\kappa}$, while
$L_{\theta}L_{\theta^{\prime}}=0$ for any $L_{\theta}\in\mathcal{\hat{E}%
}_{\kappa}$ and $\,L_{\theta^{\prime}}\in\mathcal{\hat{E}}_{\kappa^{\prime}}$
. Then,
\[
\left[  \left[  \mathcal{\hat{E}}\right]  \right]  =\bigoplus_{\kappa}\left[
\left[  \mathcal{\hat{E}}_{\kappa}\right]  \right]  .
\]
It is easy to check that $\left[  \left[  \mathcal{\hat{E}}_{\kappa}\right]
\right]  $ does not break into direct sum of smaller subalgebras. In fact, by
the first condition, if $\left\vert u_{\theta}\right\rangle \left\langle
u_{\theta}\right\vert $, $\left\vert u_{\theta^{\prime}}\right\rangle
\left\langle u_{\theta^{\prime}}\right\vert $ are the elements of $\left[
\left[  \mathcal{\hat{E}}_{\kappa}\right]  \right]  $, so are $\left\vert
u_{\theta}\right\rangle \left\langle u_{\theta^{\prime}}\right\vert $ and
$\left\vert u_{\theta^{\prime}}\right\rangle \left\langle u_{\theta
}\right\vert $. Thus, $\left[  \left[  \mathcal{\hat{E}}_{\kappa}\right]
\right]  $ is nothing but the linear operators on the space spanned by
$\mathcal{U}=\left\{  u_{\theta}\right\}  _{\theta\in\Theta}$.

\begin{proposition}
Suppose that $L_{\theta}=\left\vert u_{\theta}\right\rangle \left\langle
u_{\theta}\right\vert $, $M_{\theta}$ is a constant multiple of projector, and
$\left\Vert L_{\theta}\right\Vert =\left\Vert M_{\theta}\right\Vert $ for each
$\theta$. Then, there is a CP subunital map $\Lambda$ satisfying
(\ref{convert}) exists if and only if there is a system $\mathcal{V}=\left\{
v_{\theta}\right\}  _{\theta\in\Theta}$ of vectors such that $\mathcal{V}$ is
unitary equivalent to $\mathcal{U}=\left\{  u_{\theta}\right\}  _{\theta
\in\Theta}$ and that $M_{\theta}=\left\vert v_{\theta}\right\rangle
\left\langle v_{\theta}\right\vert \otimes I_{d_{\kappa}}$ ($L_{\theta}%
\in\mathcal{\hat{E}}_{\kappa}$).
\end{proposition}

When $L_{\theta}$ is also not of rank 1, still one can state something. The
proof of the following proposition is straightforward, thus omitted.

\begin{proposition}
Suppose that $L_{\theta}$, $M_{\theta}$ is a constant multiple of projector,
and $\left\Vert L_{\theta}\right\Vert =\left\Vert M_{\theta}\right\Vert $ for
each $\theta$. Then, there is a CP subunital map $\Lambda$ satisfying
(\ref{convert}) exists if and only if
\begin{align*}
L_{\theta}  &  =U_{1}\,\bigoplus_{n\in\mathbb{N}}a_{n,\theta}I_{n}\otimes
I_{d_{1,n}}U_{1}^{\dagger},\\
M_{\theta}  &  =U_{2}\,\bigoplus_{n\in\mathbb{N}}b_{n,\theta}I_{n}\otimes
I_{d_{2,n}}U_{2}^{\dagger},
\end{align*}
where $U_{1}$ , $U_{2}$ are unitary, and $b_{n,\theta}\neq0$ and $d_{2,n}%
\neq0$ unless $a_{n,\theta}\neq0$ or $d_{1,n}=0$.
\end{proposition}

\bigskip

\section{From rank-1 operators to arbitrary operators}

\begin{theorem}
Suppose $\mathcal{U}=\left\{  u_{\theta}\right\}  _{\theta\in\Theta}$ is
linearly independent. Then there is a CP subunital map $\Lambda$ with
(\ref{convert}) exists if and only if there are operators $M_{\theta
,\theta^{\prime}}\in\mathcal{L}\left(  \mathcal{K}\right)  $ ($\theta$%
,$\theta^{\prime}\in\Theta$) such that
\begin{align}
\sum_{\theta,\theta^{\prime}\in\Theta}M_{\theta,\theta^{\prime}}%
\otimes\left\vert e_{\theta}\right\rangle \left\langle e_{\theta^{\prime}%
}\right\vert  &  \geq0,\,M_{\theta,\theta}=M_{\theta}\label{M-1}\\
\sum_{\theta,\theta^{\prime}\in\Theta}M_{\theta,\theta^{\prime}}\left(
G_{\mathcal{U}}^{-1}\right)  _{\theta,\theta^{\prime}}  &  \leq I_{\mathcal{K}%
}, \label{M-2}%
\end{align}
where $\left\{  e_{\theta}\right\}  _{\theta\in\Theta}$ is an orthonormal
system of vectors spanning a Hilbert space $\mathcal{K}^{\prime}$.

Also, there is a CP unital map $\Lambda$ with (\ref{convert}) exists if and
only if there are operators $M_{\theta,\theta^{\prime}}\in\mathcal{L}\left(
\mathcal{K}\right)  $ ($\theta$,$\theta^{\prime}\in\Theta$) with (\ref{M-1})
and
\begin{equation}
\sum_{\theta,\theta^{\prime}\in\Theta}M_{\theta,\theta^{\prime}}\left(
G_{\mathcal{U}}^{-1}\right)  _{\theta,\theta^{\prime}}=I_{\mathcal{K}}.
\label{M-3}%
\end{equation}

\end{theorem}

\begin{proof}
$\Lambda$ is CP if and only if
\[
\sum_{\theta,\theta^{\prime}\in\Theta}\Lambda\left(  \left\vert u_{\theta
}\right\rangle \left\langle u_{\theta^{\prime}}\right\vert \right)
\otimes\left\vert e_{\theta}\right\rangle \left\langle e_{\theta^{\prime}%
}\right\vert =\sum_{\theta,\theta^{\prime}\in\Theta}M_{\theta,\theta^{\prime}%
}\otimes\left\vert e_{\theta}\right\rangle \left\langle e_{\theta^{\prime}%
}\right\vert \geq0,
\]
where we put \
\[
M_{\theta,\theta^{\prime}}=\Lambda\left(  \left\vert u_{\theta}\right\rangle
\left\langle u_{\theta^{\prime}}\right\vert \right)  .
\]

Observe
\begin{align*}
\sum_{\theta,\theta^{\prime}\in\Theta}\Lambda\left(  \left\vert u_{\theta
}\right\rangle \left\langle u_{\theta^{\prime}}\right\vert \right)
\otimes\left\vert e_{\theta}\right\rangle \left\langle e_{\theta^{\prime}%
}\right\vert  &  =\sum_{\theta,\theta^{\prime}\in\Theta}\Lambda^{\ast}\left(
\left[  \mathcal{U}\right]  \left\vert e_{\theta}\right\rangle \left\langle
e_{\theta^{\prime}}\right\vert \left[  \mathcal{U}\right]  ^{\dagger}\right)
\otimes\left\vert e_{\theta}\right\rangle \left\langle e_{\theta^{\prime}%
}\right\vert \\
&  =\sum_{\theta,\theta^{\prime}\in\Theta}\Lambda^{\ast}\left(  \left\vert
e_{\theta}\right\rangle \left\langle e_{\theta^{\prime}}\right\vert \right)
\otimes\left[  \mathcal{U}\right]  ^{T}\left\vert e_{\theta}\right\rangle
\left\langle e_{\theta^{\prime}}\right\vert \overline{\left[  \mathcal{U}%
\right]  }.
\end{align*}
Equating this with $\sum_{\theta,\theta^{\prime}\in\Theta}M_{\theta
,\theta^{\prime}}\otimes\left\vert e_{\theta}\right\rangle \left\langle
e_{\theta^{\prime}}\right\vert $ and solving about $\sum_{\theta
,\theta^{\prime}\in\Theta}\Lambda\left(  \left\vert e_{\theta}\right\rangle
\left\langle e_{\theta^{\prime}}\right\vert \right)  \otimes\left\vert
e_{\theta}\right\rangle \left\langle e_{\theta^{\prime}}\right\vert $, we
have
\[
\sum_{\theta,\theta^{\prime}\in\Theta}\Lambda\left(  \left\vert e_{\theta
}\right\rangle \left\langle e_{\theta^{\prime}}\right\vert \right)
\otimes\left\vert e_{\theta}\right\rangle \left\langle e_{\theta^{\prime}%
}\right\vert =\sum_{\theta,\theta^{\prime}\in\Theta}M_{\theta,\theta^{\prime}%
}\otimes\left(  \left[  \mathcal{U}\right]  ^{T}\right)  ^{-1}\left\vert
e_{\theta}\right\rangle \left\langle e_{\theta^{\prime}}\right\vert \left(
\overline{\left[  \mathcal{U}\right]  }\right)  ^{-1}.
\]
Therefore, \ $\Lambda$ is subunital if and only if
\begin{align*}
I_{\mathcal{K}}  &  \geq\mathrm{tr}\,_{\mathcal{K}^{\prime}}\sum
_{\theta,\theta^{\prime}\in\Theta}\Lambda\left(  \left\vert e_{\theta
}\right\rangle \left\langle e_{\theta^{\prime}}\right\vert \right)
\otimes\left\vert e_{\theta}\right\rangle \left\langle e_{\theta^{\prime}%
}\right\vert \\
&  =\sum_{\theta,\theta^{\prime}\in\Theta}M_{\theta,\theta^{\prime}%
}\left\langle e_{\theta^{\prime}}\right\vert \left(  \overline{\left[
\mathcal{U}\right]  }\right)  ^{-1}\left(  \left[  \mathcal{U}\right]
^{T}\right)  ^{-1}\left\vert e_{\theta}\right\rangle \\
&  =\sum_{\theta,\theta^{\prime}\in\Theta}M_{\theta,\theta^{\prime}}\left(
G_{\mathcal{U}}^{-1}\right)  _{\theta,\theta^{\prime}}.
\end{align*}

The map $\Lambda$ is unital if and only if the equality in the above
inequality holds. Thus we have (\ref{M-3}).
\end{proof}

\bigskip

\begin{remark}
Suppose $M_{\theta}=\left\vert v_{\theta}\right\rangle \left\langle v_{\theta
}\right\vert $, for each $\theta\in\Theta$. Then, by (\ref{Luu=Hvv}),
(\ref{M-1}) and (\ref{M-2}) become
\begin{align*}
\sum_{\theta,\theta^{\prime}\in\Theta}H_{\theta,\theta^{\prime}}\left\vert
v_{\theta}\right\rangle \left\langle v_{\theta^{\prime}}\right\vert
\otimes\left\vert e_{\theta}\right\rangle \left\langle e_{\theta^{\prime}%
}\right\vert  &  \geq0,\\
\sum_{\theta,\theta^{\prime}\in\Theta}\left\vert v_{\theta}\right\rangle
\left\langle v_{\theta^{\prime}}\right\vert H_{\theta,\theta^{\prime}}\left(
G_{\mathcal{U}}^{-1}\right)  _{\theta,\theta^{\prime}}  &  \leq I_{\mathcal{K}%
},
\end{align*}
respectively. The first inequality is verified by
\[
\sum_{\theta,\theta^{\prime}\in\Theta}H_{\theta,\theta^{\prime}}\left\vert
v_{\theta}\right\rangle \left\langle v_{\theta^{\prime}}\right\vert
\otimes\left\vert e_{\theta}\right\rangle \left\langle e_{\theta^{\prime}%
}\right\vert =A\sum_{\theta,\theta^{\prime}\in\Theta}\left\vert v_{1}%
\right\rangle \left\langle v_{1}\right\vert \otimes H_{\theta,\theta^{\prime}%
}\left\vert e_{\theta}\right\rangle \left\langle e_{\theta^{\prime}%
}\right\vert A^{\dagger}\geq0,
\]
where
\[
A:=\sum_{\theta\in\Theta}\left\vert v_{\theta}\right\rangle \left\langle
v_{1}\right\vert \otimes\left\vert e_{\theta}\right\rangle \left\langle
e_{\theta}\right\vert .
\]
The second inequality can be rewritten as
\begin{align*}
I_{\mathcal{K}}  &  \geq\sum_{\theta,\theta^{\prime}\in\Theta}\left\vert
v_{\theta}\right\rangle \left\langle v_{\theta^{\prime}}\right\vert
H_{\theta,\theta^{\prime}}\left(  G_{\mathcal{U}}^{-1}\right)  _{\theta
,\theta^{\prime}}\\
&  =\left[  \mathcal{V}\right]  H\circ\left(  G_{\mathcal{U}}^{-1}\right)
\left[  \mathcal{V}\right]  ^{\dagger}.
\end{align*}
Hence, if $\mathcal{V}=\left\{  v_{\theta}\right\}  _{\theta\in\Theta}$ is
linearly independent, we obtain (\ref{G>HG}).
\end{remark}

\bigskip

Inserting some $M_{\theta,\theta^{\prime}}$ satisfying (\ref{M-1}) into
(\ref{M-2}), one obtain sufficient condition for (\ref{convert}) to hold for a
CP subunital map $\Lambda^{\ast}$. For example,
\[
\left\Vert \sum_{\theta\in\Theta}M_{\theta}\left(  G_{\mathcal{U}}%
^{-1}\right)  _{\theta,\theta}\right\Vert \leq1,
\]
or%
\[
\left\Vert \sum_{\theta,\theta^{\prime}\in\Theta}\sqrt{M_{\theta}}%
\sqrt{M_{\theta^{\prime}}}\left(  G_{\mathcal{U}}^{-1}\right)  _{\theta
,\theta^{\prime}}\right\Vert \leq1,
\]
and so on.

But obtaining the necessary and sufficient condition for \ (\ref{M-1}) and
(\ref{M-2}), or for (\ref{M-1}) and (\ref{M-3}), is quite non-trivial task, in
general. So in the next subsection, we deal with an easy case, the case of
$\left\vert \Theta\right\vert =2$ and $\dim\mathcal{H}=\dim\mathcal{K}=2$.

\section{2-dimensional and $\left\vert \Theta\right\vert =2$ case}

In this section, we work on the case of $\left\vert \Theta\right\vert =2$ and
$\dim\mathcal{H}=\dim\mathcal{K}=2$. First, we note that the problem is reduce
to the case of $\mathrm{rank}\,L_{\theta}=1$ and $\left\Vert L_{\theta
}\right\Vert =1$ ($\theta=1,2$). Observe that (\ref{convert}) for a CP map
$\Lambda$ is equivalent to \
\begin{align*}
\Lambda\left(  L_{1}-t_{1}L_{2}\right)   &  =M_{1}-t_{1}M_{2},\\
\Lambda\left(  L_{2}-t_{2}L_{1}\right)   &  =M_{2}-t_{2}M_{1}.
\end{align*}
Choose $t_{1}$ and $t_{2}$ so that $L_{1}-t_{1}L_{2}$ and $L_{2}-t_{2}L_{1}$
is rank-1 positive operator, the problem is reduced to the case of
\ $\mathrm{rank}\,L_{\theta}=1$ ($\theta=1,2$). Since multiplying constant to
the input and the output does not change the problem, $\left\Vert L_{\theta
}\right\Vert =1$ ($\theta=1,2$) can be assumed without loss of generality, too.

\begin{remark}
In case that $\Lambda$ is unital, there is another way of reducing the problem
to the case of $\mathrm{rank}\,L_{\theta}=1$ and $\left\Vert L_{\theta
}\right\Vert =1$ ($\theta=1,2$). Observe that (\ref{convert}) for a CP
\textit{unital} map $\Lambda$ is the same as
\[
\Lambda\left(  a_{\theta}\left(  L_{\theta}-b_{\theta}I_{\mathcal{H}}\right)
\right)  =a_{\theta}\left(  M_{\theta}-b_{\theta}I_{\mathcal{K}}\right)  ,
\]
due to the fact that $\Lambda$ is linear and unital. Therefore, choosing
$b_{\theta}:=\left\Vert L_{\theta}\right\Vert $, we can reduce the problem to
the case where $L_{\theta}=\left\vert u_{\theta}\right\rangle \left\langle
u_{\theta}\right\vert $ for each $\theta\in\Theta$. Note here that by
(\ref{infty-monotone}) $M_{\theta}-\left\Vert L_{\theta}\right\Vert
I_{\mathcal{K}}$ is positive definite. Also, by choosing $a_{\theta}$
properly, we can assume, without loss of generality, $\left\Vert
u_{1}\right\Vert =\left\Vert u_{2}\right\Vert =1$.
\end{remark}

Also, $\left\{  u_{1},u_{2}\right\}  $ is assumed to be linearly independent,
since otherwise the problem becomes trivial. \ We choose the phase of
$u_{\theta}$ so that $\left\langle u_{1}\right\vert \left.  u_{2}\right\rangle
=c\geq0$. Also, we suppose $M_{1}>0$ . (when both of $M_{1}$ and $M_{2}$ are
not strictly positive, they are rank-1, and thus reduce to the case of
$\mathrm{rank}\,M_{\theta}=1$.)

Below, we study the condition that (\ref{convert}) holds for a CP unital map
$\Lambda$, or rewrite (\ref{M-1}) and (\ref{M-3}) in more `tractable' expression.

(\ref{M-3}) is the same as \ %

\[
M_{1}+M_{2}-c\left(  M_{1,2}+M_{2,1}\right)  =\left(  1-c^{2}\right)
I_{\mathcal{K}},
\]
or equivalently,
\[
M_{1}^{-1/2}M_{2,1}M_{1}^{-1/2}=M_{0}+\sqrt{-1}B,
\]
where
\[
M_{0}:=\frac{1}{2c}M_{1}^{-1/2}\left(  M_{1}+M_{2}-\left(  1-c^{2}\right)
I_{\mathcal{K}}\right)  M_{1}^{-1/2}%
\]
and $B$ is a Hermitian operator. Therefore, (\ref{M-1}) can be rewritten as
\begin{align*}
\left(  M_{0}+\sqrt{-1}B\right)  \left(  M_{0}-\sqrt{-1}B\right)   &
=M_{1}^{-1/2}M_{2,1}M_{1}^{-1}M_{1,2}M_{1}^{-1/2}\\
&  \leq M_{1}^{-1/2}M_{2}M_{1}^{-1/2}.
\end{align*}
Defining
\[
\mathcal{C}:=\left\{  M\,;\,M\geq B^{2}+\sqrt{-1}\left[  B,M_{0}\right]
\right\}  ,
\]
the above inequality is equivalent to
\[
M_{1}^{-1/2}M_{2}M_{1}^{-1/2}-M_{0}^{2}\in\mathcal{C}.
\]
So our task is to find a convenient expression of the set $\mathcal{C}$.

In the previous paper, we have already have solved this problem (Appendix C of
\cite{Matsumoto}). Choose an orthonormal basis of $\mathcal{K}$ so that
$M_{0}$ defined above is diagonalized. Let us parameterize $M_{0}$ and
$M_{1}^{-1/2}M_{2}M_{1}^{-1/2}-M_{0}^{2}$ as follows,
\begin{align*}
M_{0}  &  =l\sigma_{z}+mI_{2},\\
M_{1}^{-1/2}M_{2}M_{1}^{-1/2}-M_{0}^{2}  &  =l^{2}\left(  x\sigma_{x}%
+y\sigma_{y}+z\sigma_{z}+wI_{\mathcal{K}}\right)  ,
\end{align*}
where $\sigma_{x}$, $\sigma_{y}$, and $\sigma_{z}$ are Pauli matrices. Then,
\begin{align}
&  M_{1}^{-1/2}M_{2}M_{1}^{-1/2}-M_{0}^{2}\in\mathcal{C}\nonumber\\
&  \Leftrightarrow\exists s\in\left[  -2,2\right]  \,,\,\,x^{\prime}\sigma
_{x}+z\sigma_{z}+wI_{\mathcal{K}}\geq s\sigma_{x}+\frac{s^{2}}{4}%
I_{\mathcal{K}}\nonumber\\
&  \Leftrightarrow z=0,\,w\geq f_{1}\left(  x^{\prime}\right)  \,\text{
}\nonumber\\
\,\text{or }z  &  \neq0,w\geq f_{2}\left(  x^{\prime},z\right)  ,\,\,f_{3}%
\left(  x^{\prime},z,w\right)  \geq0 \label{convert-qubit}%
\end{align}
where we have defined
\begin{align*}
x^{\prime}  &  :=\sqrt{x^{2}+y^{2}}\\
f_{1}\left(  x\right)   &  :=\left\{
\begin{array}
[c]{cc}%
\left\vert x\right\vert -1, & \left(  \left\vert x\right\vert \geq2\right) \\
\frac{1}{4}x^{2}, & \left(  \left\vert x\right\vert \leq2\right)
\end{array}
\right.  ,\\
f_{2}\left(  x,z\right)   &  :=\left\{
\begin{array}
[c]{cc}%
\sqrt{x^{2}+z^{2}}-1, & \left(  x^{2}+z^{2}\geq4\right) \\
\frac{1}{4}\left\{  x^{2}+z^{2}\right\}  , & \left(  x^{2}+z^{2}\leq4\right)
\end{array}
,\,\right.
\end{align*}
and
\begin{align}
&  f_{3}\left(  x,z,w\right) \nonumber\\
&  :=16w^{4}+\left(  -8x^{2}+8z^{2}+32\right)  w^{3}+\left(  x^{4}+2x^{2}%
z^{2}-32x^{2}+z^{4}-8z^{2}+16\right)  w^{2}\\
&  +\left(  10x^{4}+2x^{2}z^{2}-8x^{2}-8z^{4}-32z^{2}\right)  w+\left(
x^{4}-3x^{4}z^{2}-x^{6}-3x^{2}z^{4}+20x^{2}z^{2}-z^{6}-8z^{4}-16z^{2}\right)
\allowbreak.\nonumber
\end{align}

The quantities $l$, $x^{\prime}$, $z$, and $w$, which are needed to check
(\ref{convert-qubit}), can be computed directly by the following formulas:%
\begin{align*}
l  &  =\frac{1}{2}\sqrt{\left(  \mathrm{tr}\,M_{0}\right)  ^{2}-4\left(  \det
M_{0}\right)  ^{2}},\\
w  &  =\frac{1}{2l^{2}}\mathrm{tr}\,M,\\
z  &  =\frac{1}{2l^{3}}\mathrm{tr}\,\left\{  M\left(  M_{0}-\frac{1}{2}\left(
\mathrm{tr}\,M_{0}\right)  I_{\mathcal{K}}\right)  \right\}  =\frac{1}{2l^{3}%
}\left(  \mathrm{tr}\,MM_{0}-\frac{1}{2}\left(  \mathrm{tr}\,M_{0}\right)
\left(  \mathrm{tr}\,M\right)  \right)  ,\\
x^{\prime}  &  ==\sqrt{\mathrm{tr}\,\left\{  \frac{1}{l^{2}}M-\frac{z}{l}%
M_{0}+\left(  \frac{z}{2l}\left(  \mathrm{tr}\,M_{0}\right)  -w\right)
I_{\mathcal{K}}\right\}  ^{2}},
\end{align*}
where $M:=M_{1}^{-1/2}M_{2}M_{1}^{-1/2}-M_{0}^{2}$.

\section{Randomization criteria}

In this section, $\Theta$ is any set.

\begin{lemma}
\label{lem:Fan-minimax}(Fan's minimax theorem, \cite{BorweinZhu} ) Suppose
that $\mathcal{X}$ be a compact convex subset of vector space, and
$\mathcal{Y}$ be a convex subset of a vector space. \ Assume that
$f:\mathcal{X}\times\mathcal{Y}\rightarrow%
%TCIMACRO{\U{211d} }%
%BeginExpansion
\mathbb{R}
%EndExpansion
$ satisfies following conditions: (1) $x\rightarrow f\left(  x,y\right)  $ is
lower semi continuous and convex on $\mathcal{X}$ for every $y\in\mathcal{Y}$:
(2) $y\rightarrow f\left(  x,y\right)  $ is concave on $\mathcal{Y}$ for every
$x\in\mathcal{X}$. Then%
\[
\min_{x\in\mathcal{X}}\sup_{y\in\mathcal{Y}}f\left(  x,y\right)  =\sup
_{y\in\mathcal{Y}}\min_{x\in\mathcal{X}}f\left(  x,y\right)  .
\]

\end{lemma}

\begin{theorem}
\label{th:randomization}(randomization criteria)Let $e_{\theta}\geq0$,
$\theta\in\Theta$. There is a CP (sub)unital map $\Lambda$ satisfying \
\begin{equation}
\forall\theta\in\Theta,\,\,\left\Vert \Lambda\left(  L_{\theta}\right)
-M_{\theta}\right\Vert \leq e_{\theta} \label{convert-e}%
\end{equation}
(\ref{convert}) if and only if
\begin{equation}
\inf_{\Lambda_{1}}\,\sum_{\theta\in\Theta_{0}}p_{\theta}\mathrm{tr}%
\,\Lambda_{1}\left(  L_{\theta}\right)  X_{\theta}\leq\inf_{\Lambda_{2}}%
\,\sum_{\theta\in\Theta_{0}}\,p_{\theta}\mathrm{tr}\,\Lambda_{2}\left(
M_{\theta}\right)  X_{\theta}+e_{\theta} \label{randomization}%
\end{equation}
holds for any subset $\Theta_{0}$ of $\Theta$ with $\left\vert \Theta
_{0}\right\vert <\infty$, any probability distribution $\left\{  p_{\theta
}\right\}  _{\theta\in\Theta_{0}}$ on $\Theta_{0}$, and any family of
operators $\left\{  X_{\theta}\right\}  _{\theta\in\Theta}$ on $\mathcal{K}$
with $\left\Vert X_{\theta}\right\Vert _{1}\leq1$, $\forall\theta\in\Theta$.
Here, $\Lambda_{1}$ and $\Lambda_{2}$ moves over the set of all CP (sub)unital
maps from $\mathcal{L}\left(  \mathcal{H}\right)  $ to $\mathcal{L}\left(
\mathcal{K}\right)  $ and $\mathcal{L}\left(  \mathcal{K}\right)  $ to
$\mathcal{L}\left(  \mathcal{K}\right)  $, respectively. This in turn
equivalent to that
\begin{equation}
\,\,\sup_{\Lambda_{1}}\sum_{\theta\in\Theta_{0}}p_{\theta}\mathrm{tr}%
\,\Lambda_{1}\left(  L_{\theta}\right)  X_{\theta}+e_{\theta}\geq\sup
_{\Lambda_{2}}\,\sum_{\theta\in\Theta_{0}}p_{\theta}\mathrm{tr}\,\Lambda
_{2}\left(  M_{\theta}\right)  X_{\theta} \label{randomization-2}%
\end{equation}
holds for any subset $\Theta_{0}$ of $\Theta$ with $\left\vert \Theta
_{0}\right\vert <\infty$, any probability distribution $\left\{  p_{\theta
}\right\}  _{\theta\in\Theta_{0}}$ on $\Theta_{0}$, and any family of
operators $\left\{  X_{\theta}\right\}  _{\theta\in\Theta}$ on $\mathcal{K}$
with $\left\Vert X_{\theta}\right\Vert _{1}\leq1$, $\forall\theta\in\Theta$.
\end{theorem}

\begin{proof}
Since 'only if' part of the statement is trivial, we prove 'if' part. Let
\[
f_{1}\left(  \Lambda_{1},\left\{  X_{\theta}\right\}  \right)  :=\int
\,\mathrm{tr}\,\Lambda_{1}\left(  L_{\theta}\right)  X_{\theta}\,\mathrm{d}%
p\left(  \theta\right)  -\int\,\mathrm{tr}\,M_{\theta}X_{\theta}%
\mathrm{d}p\left(  \theta\right)  ,
\]
where $p$ is a measure whose support is with finite cardinality (The support
of $p$ is $\Theta_{0}$). Obviously, $f_{1}$ is bilinear and continuous. The
set of all CP (sub)unital maps is obviously compact, since $\mathcal{H}$ and
$\mathcal{K}$ are finite dimensional (even if they are infinite dimensional
separable Hilbert space, the set is compact with respect to a topology which
makes $f_{1}$ continuous in $\Lambda_{1}$. )

Also, by the assumption, $\min_{\Lambda_{1}}\,f_{1}\left(  \Lambda
_{1},\left\{  X_{\theta}\right\}  \right)  \leq0$ for each $\left\{
X_{\theta}\right\}  _{\theta\in\Theta}$ (let $\Lambda_{2}$ be the identity map
in (\ref{randomization})). Therefore, by Lemma\thinspace\ref{lem:Fan-minimax}%
,
\begin{align*}
e_{\theta}  &  \geq\sup\left\{  \min_{\Lambda_{1}}\,f_{1}\left(  \Lambda
_{1},\left\{  X_{\theta}\right\}  \right)  ;\,\left\Vert X_{\theta}\right\Vert
_{1}\leq1\right\} \\
&  =\min_{\Lambda_{1}}\sup\left\{  \,f_{1}\left(  \Lambda_{1},\left\{
X_{\theta}\right\}  \right)  ;\,\left\Vert X_{\theta}\right\Vert _{1}%
\leq1\right\} \\
&  =\min_{\Lambda_{1}}\sup\left\{  \sum_{\theta\in\text{\textrm{supp}}%
\,p}\mathrm{tr}\,\left(  \Lambda_{1}\left(  L_{\theta}\right)  -M_{\theta
}\right)  X_{\theta}\,;\left\Vert X_{\theta}\right\Vert _{1}\leq1\right\} \\
&  =\min_{\Lambda_{1}}\int\left\Vert \Lambda_{1}\left(  L_{\theta}\right)
-M_{\theta}\right\Vert \mathrm{d}p\left(  \theta\right)  .
\end{align*}
Next, let
\[
f_{2}\left(  \Lambda_{1},p\right)  :=\int\left\Vert \Lambda_{1}\left(
L_{\theta}\right)  -M_{\theta}\right\Vert \mathrm{d}p\left(  \theta\right)  ,
\]
which is lower semi-continuous, convex in $\Lambda_{1}$, and linear in $p$. By
Lemma\thinspace\ref{lem:Fan-minimax},%
\begin{align*}
e_{\theta}  &  \geq\sup_{p:\left\vert \text{\textrm{supp}}\,p\right\vert
<\infty}\min_{\Lambda_{1}}\int\left\Vert \Lambda_{1}\left(  L_{\theta}\right)
-M_{\theta}\right\Vert \mathrm{d}p\left(  \theta\right) \\
&  =\min_{\Lambda_{1}}\sup_{p:\left\vert \text{\textrm{supp}}\,p\right\vert
<\infty}\int\,\left\Vert \Lambda_{1}\left(  L_{\theta}\right)  -M_{\theta
}\right\Vert \mathrm{d}p\left(  \theta\right) \\
&  =\min_{\Lambda_{1}}\sup_{\theta\in\Theta}\,\left\Vert \Lambda_{1}\left(
L_{\theta}\right)  -M_{\theta}\right\Vert ,
\end{align*}
which means existence of (sub)unital map satisfying (\ref{convert-e}).
\end{proof}

\section{Commutative case}

In this section we suppose that
\begin{equation}
L_{\theta}=\sum_{i=1}^{d_{1}}l_{\theta,i}\left\vert e_{i}\right\rangle
\left\langle e_{i}\right\vert ,\,M_{\theta}=\sum_{i=1}^{d_{2}}m_{\theta
,i}\left\vert f_{i}\right\rangle \left\langle f_{i}\right\vert ,
\label{classical}%
\end{equation}
where $\left\{  e_{i}\right\}  $ and $\left\{  f_{i}\right\}  $ are a complete
orthonormal basis of $\mathcal{H}$ and $\mathcal{K}$, respectively.

\begin{theorem}
Suppose (\ref{classical}) holds. Then there is a CP unital map $\Lambda$
satisfying (\ref{convert}) if and only if for any set $\left\{  x_{\theta
}\right\}  _{\theta\in\Theta}$ of real numbers
\begin{equation}
\forall\left\{  x_{\theta}\right\}  _{\theta\in\Theta},\,\,x_{\theta}%
\in\mathbb{R},\,\lambda_{\max}\left(  \sum_{\theta\in\Theta}x_{\theta
}L_{\theta}\right)  \geq\lambda_{\max}\left(  \sum_{\theta\in\Theta}x_{\theta
}M_{\theta}\right)  . \label{lmax>lmax}%
\end{equation}
This is equivalent to
\begin{equation}
\forall\left\{  x_{\theta}\right\}  _{\theta\in\Theta},\,\,x_{\theta}%
\in\mathbb{R},\,\,\lambda_{\min}\left(  \sum_{\theta\in\Theta}x_{\theta
}L_{\theta}\right)  \leq\lambda_{\min}\left(  \sum_{\theta\in\Theta}x_{\theta
}M_{\theta}\right)  . \label{lmin<lmin}%
\end{equation}
Also, there is a CP subunital map satisfying (\ref{convert}) if and only if
\begin{equation}
\forall\left\{  x_{\theta}\right\}  _{\theta\in\Theta},\,\,x_{\theta}%
\in\mathbb{R},\,\max\left\{  \lambda_{\max}\left(  \sum_{\theta\in\Theta
}x_{\theta}L_{\theta}\right)  ,0\right\}  \geq\max\left\{  \lambda_{\max
}\left(  \sum_{\theta\in\Theta}x_{\theta}L_{\theta}\right)  ,0\right\}  ,
\label{lmax>lmax-2}%
\end{equation}
or equivalently,
\begin{equation}
\forall\left\{  x_{\theta}\right\}  _{\theta\in\Theta},\,\,x_{\theta}%
\in\mathbb{R},\,\min\left\{  \lambda_{\min}\left(  \sum_{\theta\in\Theta
}x_{\theta}L_{\theta}\right)  ,0\right\}  \geq\min\left\{  \lambda_{\min
}\left(  \sum_{\theta\in\Theta}x_{\theta}L_{\theta}\right)  ,0\right\}  ,
\label{lmin<lmin-2}%
\end{equation}

\end{theorem}

\begin{proof}
By Theorem\thinspace\ref{th:randomization}, existence of CP unital map
$\Lambda$ satisfying (\ref{convert}) is equivalent to
\begin{equation}
\sup_{P}\sum_{i,j,\theta}x_{\theta,j}P_{j,i}l_{\theta,i}\geq\sup_{P}%
\sum_{i,j,\theta}x_{\theta,j}P_{j,i}m_{\theta,i}, \label{supP>supP}%
\end{equation}
where $P$ moves over the set of all column stochastic matrices, $P_{j,i}%
\geq0,\,\sum_{i}P_{j,i}=1$. Observe
\[
\sup_{P}\sum_{i,j,\theta}x_{\theta,j}P_{j,i}l_{\theta,i}=\sum_{j}\max_{i}%
\sum_{\theta}x_{\theta,j}l_{\theta,i}.
\]
Hence,
\[
\sum_{j}\lambda_{\max}\left(  \sum_{\theta\in\Theta}x_{\theta,j}L_{\theta
}\right)  \geq\sum_{j}\lambda_{\max}\left(  \sum_{\theta\in\Theta}x_{\theta
,j}M_{\theta}\right)
\]
holds for any $x_{\theta,j}$. This is equivalent to (\ref{lmax>lmax}).
(\ref{lmin<lmin}) is obtained by replacing $x_{\theta}$ by $-x_{\theta}.$

When $\Lambda$ is subunital, $P$ in (\ref{supP>supP}) is column sub
stochastic, $P_{j,i}\geq0,\,\sum_{i}P_{j,i}\leq1$. Therefore,
\[
\sup_{P}\sum_{i,j,\theta}x_{\theta,j}P_{j,i}l_{\theta,i}=\sum_{j}\max
_{i}\left\{  0,\sum_{\theta}x_{\theta,j}l_{\theta,i}\right\}  .
\]
Thus we obtain (\ref{lmax>lmax-2}). \ (\ref{lmin<lmin}) is obtained by
replacing $x_{\theta}$ by $-x_{\theta}$.
\end{proof}

Note that the condition
\begin{equation}
\forall\left\{  x_{\theta}\right\}  _{\theta\in\Theta},\,x_{\theta}%
\in\mathbb{R},\,\left\Vert \sum_{\theta\in\Theta}x_{\theta}L_{\theta
}\right\Vert \geq\left\Vert \sum_{\theta\in\Theta}x_{\theta}M_{\theta
}\right\Vert . \label{norm>norm}%
\end{equation}
is a necessary condition of (\ref{lmax>lmax}). In fact, by%

\begin{equation}
\left\Vert L\right\Vert =\max\left\{  \lambda_{\max}\left(  L\right)
,-\lambda_{\min}\left(  L\right)  \right\}  , \label{norm-lmax-lmin}%
\end{equation}
combining (\ref{lmax>lmax}) and (\ref{lmin<lmin}) leads to (\ref{norm>norm}).
Thus, (\ref{lmax>lmax}) implies (\ref{norm>norm}). But, suppose $\left\vert
\Theta\right\vert =1$ and $\lambda_{\max}\left(  L_{1}\right)  \leq
-\lambda_{\min}\left(  L_{1}\right)  $. Then $\left\Vert L_{1}\right\Vert
=-\lambda_{\min}\left(  L_{1}\right)  $, and
\begin{align*}
\left\Vert -L_{1}\right\Vert  &  =\max\left\{  \lambda_{\max}\left(
-L_{1}\right)  ,-\lambda_{\min}\left(  -L_{1}\right)  \right\} \\
&  =\max\left\{  -\lambda_{\min}\left(  L_{1}\right)  ,\lambda_{\max}\left(
L_{1}\right)  \right\} \\
&  =-\lambda_{\min}\left(  L_{1}\right)  ,
\end{align*}
both of which are not related to $\lambda_{\max}\left(  L_{1}\right)  $. Thus
if $\lambda_{\max}\left(  L_{1}\right)  \leq-\lambda_{\min}\left(
L_{1}\right)  $ and $\lambda_{\max}\left(  M_{1}\right)  \leq-\lambda_{\min
}\left(  M_{1}\right)  $, (\ref{norm>norm}) cannot be a sufficient condition.

Another necessary condition is
\begin{equation}
\forall\left\{  x_{\theta}\right\}  _{\theta\in\Theta},\,x_{\theta}%
\in\mathbb{R},\,\mathrm{sp}\left(  \sum_{\theta\in\Theta}x_{\theta}L_{\theta
}\right)  \leq\mathrm{sp}\left(  \sum_{\theta\in\Theta}x_{\theta}M_{\theta
}\right)  . \label{sp>sp}%
\end{equation}
Observe that one of $\lambda_{\max}\left(  L\right)  =\left\Vert L\right\Vert
$ or $\lambda_{\max}\left(  L\right)  =\mathrm{sp}\left(  L\right)
-\left\Vert L\right\Vert $ is always true. Therefore, the combination of
(\ref{norm>norm}) and (\ref{sp>sp}) is equivalent to (\ref{lmax>lmax}).

\bigskip

\appendix

\section{Monotone functionals}

Let $C$ be an arbitrary Hermitian operator on $\mathcal{H}$. Suppose
\begin{equation}
\lambda_{\max}\left(  C\right)  \geq0\geq\lambda_{\min}\left(  C\right)  .
\label{lamda>0>lamda}%
\end{equation}
Observe
\[
\mathrm{tr}\,\Lambda^{\ast}\left(  \rho\right)  =\mathrm{tr}\,\rho
\Lambda\left(  I_{\mathcal{H}}\right)  \leq\mathrm{tr}\,\left(  \rho\cdot
I_{\mathcal{K}}\right)  =1,
\]
holds for any CP subunital $\Lambda$. Therefore, if $\lambda_{\max}\left(
C\right)  \geq0$, we have
\begin{align*}
\lambda_{\max}\left(  \Lambda\left(  C\right)  \right)   &  =\max_{\rho
:\rho\geq0,\mathrm{tr}\,\rho=1}\mathrm{tr}\,\rho\Lambda\left(  C\right) \\
&  \leq\max_{\rho:\rho\geq0,\mathrm{tr}\,\rho\leq1}\mathrm{tr}\,\rho
\Lambda\left(  C\right) \\
&  =\max_{\rho:\rho\geq0,\mathrm{tr}\,\rho\leq1}\mathrm{tr}\,\Lambda^{\ast
}\left(  \rho\right)  C\\
&  \leq\max_{\rho:\rho\geq0,\mathrm{tr}\,\rho\leq1}\mathrm{tr}\,\rho C\\
&  =\max_{\rho:\rho\geq0,\mathrm{tr}\,\rho=1}\mathrm{tr}\,\rho C\\
&  =\lambda_{\max}\left(  C\right)  .
\end{align*}
and
\begin{align*}
\lambda_{\min}\left(  \Lambda\left(  C\right)  \right)   &  =\min_{\rho
:\rho\geq0,\mathrm{tr}\,\rho=1}\mathrm{tr}\,\rho\Lambda\left(  C\right) \\
&  =\min_{\rho:\rho\geq0,\mathrm{tr}\,\rho=1}\mathrm{tr}\,\Lambda^{\ast
}\left(  \rho\right)  C\\
&  \geq\min_{\rho:\rho\geq0,\mathrm{tr}\,\rho=1}\mathrm{tr}\,\rho C\\
&  \geq\lambda_{\min}\left(  C\right)  .
\end{align*}

If $\Lambda$ is CP and unital, these inequality holds without the restriction
(\ref{lamda>0>lamda}), due to almost parallel argument.

Since $\left\Vert C\right\Vert =\max\left\{  \lambda_{\max}\left(  C\right)
,-\lambda_{\min}\left(  C\right)  \right\}  $ for a Hermitian operator $C$,
the inequality
\begin{equation}
\left\Vert C\right\Vert \geq\left\Vert \Lambda\left(  C\right)  \right\Vert
\label{infty-monotone}%
\end{equation}
holds for any CP subunital map $\Lambda$ and any Hermitian operator $C$
without the restriction (\ref{lamda>0>lamda}).

\section{Multiplicative domain and finite dimensional $C^{\ast}$-algebra}

This section is based on Section 3 of \cite{HMPB}. When $\Lambda$ is
completely positive from $\mathcal{L}\left(  \mathcal{H}\right)  $ to
$\mathcal{L}\left(  \mathcal{K}\right)  $, we have Schwartz inequality
\begin{equation}
\Lambda\left(  L^{\dagger}\right)  \Lambda\left(  L\right)  \leq\left\Vert
\Lambda\right\Vert \Lambda\left(  L^{\dagger}L\right)  . \label{schwartz}%
\end{equation}
The multiplicative domain $\mathcal{M}_{\Lambda}$ of of completely positive
map $\Lambda$ is a set of operators on $\mathcal{H}$ such that
\[
\mathcal{M}_{\Lambda}:=\left\{  \Lambda\left(  L^{\dagger}\right)
\Lambda\left(  L\right)  =\left\Vert \Lambda\right\Vert \Lambda\left(
L^{\dagger}L\right)  \right\}  .
\]

\begin{remark}
When $\Lambda$ is a positive map which may not be 2-positive, we have to
replace $\left\Vert \Lambda^{\ast}\right\Vert $ in the above expressions with
$\left\Vert \Lambda^{\ast}\right\Vert _{S}$ defined by
\[
\left\Vert \Lambda\right\Vert _{S}:=\inf\left\{  c\,;\Lambda\left(
L^{\dagger}\right)  \Lambda\left(  L\right)  \leq c\Lambda\left(  L^{\dagger
}L\right)  ,L\in\mathcal{L}\left(  \mathcal{H}\right)  \right\}
\]
But as is remarked in Section 3 of \cite{HMPB}, $\left\Vert \Lambda\right\Vert
$ $=\left\Vert \Lambda\right\Vert _{S}$ when $\Lambda$ is 2-positive.
\end{remark}

\begin{lemma}
\label{lem:multiplicative}(Lemma$\,$3.9 of \cite{HMPB})Let $\Lambda$ be a
completely positive map from $\mathcal{L}\left(  \mathcal{H}\right)  $ to
$\mathcal{L}\left(  \mathcal{K}\right)  $. $\mathcal{M}_{\Lambda}$ is a vector
space closed by multiplication and $\dagger$, or constitutes a *-algebra.
Also, if $L\in\mathcal{M}_{\Lambda^{\ast}}$, for any $M\in\mathcal{L}\left(
\mathcal{H}\right)  $%
\[
\Lambda\left(  L\right)  \Lambda\left(  M\right)  =\left\Vert \Lambda
\right\Vert \Lambda\left(  LM\right)  .
\]

\end{lemma}

\begin{lemma}
\label{lem:L-norm}(Corollary 2.3.8 of \cite{Bhatia})Let $\Lambda$ be a
positive map from $\mathcal{L}\left(  \mathcal{H}\right)  $ to $\mathcal{L}%
\left(  \mathcal{K}\right)  $. Then, $\left\Vert \Lambda\right\Vert
=\left\Vert \Lambda\left(  I_{\mathcal{H}}\right)  \right\Vert $.
\end{lemma}

\begin{lemma}
\label{lem:finite-dim}(Theorem\thinspace III.1.1 and Corollary III.2.1 of
\cite{Davidson}) Any finite dimensional $C^{\ast}$-algebra is *-isomorphic to
\[
\mathcal{L}\left(  \mathbb{C}^{n_{1}}\right)  \oplus\cdots\oplus
\mathcal{L}\left(  \mathbb{C}^{n_{k}}\right)  .
\]
Also, If $\pi$ is a non-degenerate *-representation of a finite dimensional
$C^{\ast}$-algebra above, then there are cardinal numbers $d_{1}$,$\cdots$,
$d_{k}$ so that it is unitarily equivalent to $\mathbf{I}_{\mathbb{C}^{n_{1}}%
}^{\left(  d_{1}\right)  }\oplus\cdots\oplus\mathbf{I}_{\mathbb{C}^{n_{k}}%
}^{\left(  d_{k}\right)  }$ .
\end{lemma}

\end{document}